\documentclass[12pt,a4paper]{article}

\usepackage{amsfonts, amsmath, amsthm, amssymb}
\usepackage{enumerate}
\usepackage{epsfig} 
\usepackage{graphicx}
\usepackage[T1]{fontenc} 
\usepackage[english]{babel}
\usepackage[applemac]{inputenc}

\usepackage{dsfont}

\usepackage[arrow, matrix, curve]{xy}

\newcommand{\M}{\mathcal P}
\newcommand{\IR}{\mathbb R}
\newcommand{\IP}{\mathbb P}
\newcommand{\IE}{\mathbb E}
\newcommand{\IN}{\mathbb N}

\newcommand{\dd}{\mathrm{d}}
\newcommand{\dbl}[1]{d_{BL}(#1)}

\newcommand{\supp}{\mathrm{supp}\,}

\theoremstyle{plain}
\newtheorem{Theorem}{Theorem}[section]
\newtheorem{Lemma}[Theorem]{Lemma}

\newtheorem{Proposition}[Theorem]{Proposition}
\theoremstyle{definition}
\newtheorem{Definition}[Theorem]{Definition}
\newtheorem{Remark}[Theorem]{Remark}

\newtheorem{Assumption}[Theorem]{Assumption}

\begin{document}
\title{The Vlasov-Poisson dynamics as the mean field limit of extended charges}
\author{Dustin Lazarovici\thanks{lazarovici@math.lmu.de} \\[1.5ex]
Mathematisches Institut, Ludwig-Maximilians Universit\"at\\ 
Theresienstr. 39, 80333 Munich, Germany.}
\maketitle
\abstract{\noindent The paper treats the validity problem of the nonrelativistic Vlasov-Poisson equation in $d\geq 2$ dimensions. It is shown that the Vlasov-Poisson dynamics can be derived as a combined mean field and point-particle limit of an N-particle Coulomb system of extended charges. This requires a sufficiently fast convergence of the initial empirical distributions. If the electron radius decreases slower than $N^{-\frac{1}{d(d+2)}}$, the corresponding initial configurations are typical. This result entails propagation of molecular chaos for the respective dynamics.}

\section{Introduction}
\noindent We are interested in a microscopic derivation of the Vlasov-Poisson dynamics in $d \geq 2$ spatial dimensions. This is the system of equations
\begin{align}\label{VP}
&\partial_t f + p \cdot \nabla_q f + (k* \rho_t) \cdot \nabla_p f = 0, 
\end{align}

\noindent where $k$ is the Coulomb kernel

\begin{equation}\label{Coulombkernel} k(q) := \sigma\frac{q}{\,\lvert q \rvert ^{d}},\hspace{4mm} \sigma=\lbrace \pm 1 \rbrace \end{equation}
and
\begin{equation}  \rho_t(q) =  \rho[f_t](q) = \int  f(t,q,p) \, \mathrm{d}^3p  \end{equation}
\noindent is the charge density induced by the distribution function $f(t,p,q) \geq 0$, describing the density of particles with position $q \in \IR^d$ and momentum $p \in \IR^d$. Here, units are chosen such that all constants, in particular the mass and charge of the particles, are equal to $1$. 

The Vlasov-Poisson equation provides an effective description of a collisionless plasma with electrostatic ($\sigma = +1$) or gravitational  ($\sigma=-1$) interactions. In the gravitational case, the equation is also known as Vlasov-Newton.\\

\noindent \textbf{Derivation of mean field equations.} Kinetic equations of the Vlasov type are usually conceived of as \textit{mean field equations}, effective descriptions of many-particle systems, in which the $N$-particle interactions are approximated by an ``average'' effect, determining an autonomous time-evolution for the distribution function $f$.

Classical results, dealing with simplified models with Lipschitz-continuous forces, prove a statement of the following kind: If for an initial microscopic configuration $Z=(q_i,p_i)_{i=1,..,N}$ the empirical distribution $\mu^N_0[Z] = \frac{1}{N} \sum\limits_{i=1}^N \delta_{q_i}\delta_{p_i}$ approximates a continuous density $f_0$, then, at time $t>0$, the time-evolved distribution $\mu^N_t= \frac{1}{N} \sum\limits_{i=1}^N \delta_{q_i(t)}\delta_{p_i(t)}$ approximates $f_t$, where $f_t$ is a solution of the corresponding Vlasov equation of the form \eqref{VP}. Formally, the approximation is understood in terms of weak convergence of probability measures, quantified by an appropriate metric. (Neunzert and Wick, 1974 \cite{NeunzertWick}, Braun and Hepp, 1977 \cite{BraunHepp}, Dobrushin, 1979 \cite{Dobrushin}; see also \cite{Neunzert, SpohnBook}.)\\

\noindent \textbf{Results for singular forces}.  For singular forces -- up to but not including the Coulomb case --  similar results could be proven only recently. Hauray and Jabin (2013) \cite{HaurayJabin} treat force kernels bounded as $\lvert k(q) \rvert \leq \frac{C}{\lvert q \rvert^\alpha}$ with $\alpha < d-1$. For $1 < \alpha < d-1$ they require an $N$-dependent cut-off which can be chosen as small as $N^{-1/2d}$ for $\alpha \nearrow d-1$, while for $\alpha < 1$ they are able to perform the mean field limit without cut-off. Pickl and Boers (2015) improve the cut-off for singularities near the Coulomb case to $N^{-1/d}$ \cite{Peter}. 

The microscopic justification of the Vlasov-Poisson equation, corresponding to the case $\alpha = d-1$, has been an open problem, so far. In this paper, we propose a particle approximation by \emph{extended charges} with $N$-dependent radius that can decrease as fast as $N^{-\frac{1}{d(d+2)} + \epsilon}$. The proof is based on a stability result of Loeper \cite{Loeper} and an anisotropic variant of the Wasserstein distance. An alternative proof, generalization the methods introduced in \cite{Peter}, is simultaneously presented in \cite{PeterDustin}. 

While the proof in \cite{PeterDustin} requires a much smaller cut-off of order $N^{-\frac{1}{d}+\epsilon}$, the method presented here allows for better rates of convergence. Moreover, the microscopic regularization proposed here can be understood as a nonrelativistic analogue of the rigid charges model that was used by Golse (2012) to perform the mean field limit for a regularized version of the Vlasov-Maxwell dynamics \cite{Golse}. Our discussion might thus also be interesting in view of a possible generalization to the relativistic Vlasov-Maxwell system.\\

\noindent \textbf{Convergence in law and molecular chaos.}
What all recent results with singular forces have in common, is that they are probabilistic in the sense that the mean field limit can be performed for \emph{typical} initial conditions. In other words, the microscopic density $\mu^N_t$ converges in law to the constant variable $f_t$, which is given as the solution of the corresponding mean field equation. By a well-known result in probability theory (e.g. \cite{Kac}, \cite{Grunbaum}, \cite[Prop.2.2 ]{Sznitman}, \cite{Mischler}), this is equivalent to \emph{molecular chaos} in the following sense: 
If at time $t=0$ the particles are identically and independently distributed with law $f_0$ and if the corresponding product measure $F^N_0 = \otimes^N f_0$ on $\IR^{6N}$ evolves with the microscopic $N$-particle flow, then, for times $t>0$, it  holds that $F^N_t = \Phi_t^N \# F^N_0 \approx \otimes^N f_t$, where the approximation is understood in terms of convergence of marginals. That is, writing $x_i=(q_i, p_i)$, we consider the $k$-particle marginal
\begin{equation*}\label{molecularchaos} ^{(k)}{F^N_t(x_1,...,x_k)} := \int F^N_t(Z) \, \mathrm{d}x_{k+1}... \mathrm{d}x_{N}.\end{equation*}
Then $^{(k)}{F^N_t}$ converges weakly to $\otimes^k f_0$ as $N \to \infty$ for all $k \in \IN$. 

\noindent Note that the probabilistic nature of these statements is in contrast to earlier results in \cite{BraunHepp} and \cite{Dobrushin}, which are, in effect, deterministic, allowing arbitrary sequences of initial configurations. One reason is that for singular forces, there exist ``bad'' initial conditions leading to clustering of particles and significant deviations from the typical mean field behavior.\\

\noindent \textbf{Comparison of recent result, open problems} 
The strategy employed in \cite{HaurayJabin}, as well as in the present paper, is thus to impose additional constraints on the initial conditions, subsequently showing that these constraints are satisfied with probability $1$ in the limit $N \to \infty$. In \cite{HaurayJabin}, the respective bounds are imposed on the concentration of particles at $t=0$, while in our proof, the probabilistic element enters through the requirement of a sufficiently fast convergence of the initial microscopic distribution. In any case, these assumptions assure that the initial configuration is ``well-placed'', so to speak, preventing, in particular, a blow-up of the microscopic dynamics. 

One of the key innovations in our proof is that a regularization is applied on the level of the charge density, which allows us to work with bounded densities rather than Dirac masses. The $L^\infty$-norm of the microscopic charge density is controlled by propagating a respective Wasserstein bound (Lemma \ref{Lemma:rhobound}). Similar estimates can be used to carry over stronger regularity properties from the Vlasov density to the microscopic density.  While this method is rather simple, it requires a relatively large cut-off of order $\sim N^{-\frac{1}{d(d+2)} + \epsilon}$ for some $\epsilon > 0$. Moreover, there is no immediate connection between the width of the cut-off and the strength of the singularity. This is in contrast to the situation in \cite{HaurayJabin} or \cite{Peter}, where the lower bound on the required regularization decreases with $\alpha$. 

Most notably, Hauray and Jabin, 2013 \cite{HaurayJabin}, are able to prove molecular chaos for weak singularities ($\alpha < 1$) with no cut-off at all, while the method proposed here requires in any case a regularization (smearing of the charges). However, we emphasize that the results in \cite{HaurayJabin} do not include the Coulomb case $\alpha = d-1$, which is the main focus of our paper. Furthermore, our result applies also in dimension 2, while the assumptions required in \cite{HaurayJabin} are no longer generic in that case.

The method introduced in \cite{Peter} and extended in \cite{PeterDustin} is designed for stochastic initial conditions, thus aiming directly at a typicality result. Rather than controlling the difference between $f^N$ and $\mu^N$ in some weak metric, one considers a stochastic process of the form $\IE(\lvert \Phi^f_{t,0}(Z) - \Phi^\mu_{t,0}(Z) \rvert_\infty)$, where  $ \Phi^f_{t,0}, \Phi^\mu_{t,0} $ are the N-particle flows generated by the mean field dynamics and the microscopic dynamics, respectively. The corresponding proof allows the cut-off to decrease as fast as $N^{-1/d+\epsilon}$, i.e. (almost) as fast as the typical distance between a particle and its nearest neighbour. 

Whether mean field results for strongly singular forces -- approaching or even including the Coulomb case -- can be obtained with no cut-off at all is an open question. Concerning related problems, we believe that some of the methods presented here could be generalized to the relativistic Vlasov-Maxwell dynamics and leave this to be treated in a future paper.

\section{The microscopic model}

As the force kernel considered here is strongly singular at the origin, we will require a regularization on the microscopic level. We shall consider as a microscopic model the dynamics of smeared (extended) charges with Coulomb interactions. The cut-off parameter $r_N$ thus has a straight-forward physical interpretation as a finite electron radius. In the relativistic case, an analogous model of rigid charges (without collisions and rotations) was used by Golse to derive a regularized version of the Vlasov-Maxwell system \cite{Golse} (c.f. also Rein, 2004 \cite{Rein2}).

While the smearing of charges is a natural way to regularize point-interactions, the cut-off thus introduced is a technical necessity rather than a realistic physical model. In particular, the $N$-dependence of the electron radius might seem strange from a physical point of view, though similar regularizations are commonly used in numerical simulations. Also note that in our proof, the radius has to be chosen so large that a great number of particles will typically overlap. Intuitively, the combined limit $N\to \infty, \; r_N\to 0$ describes a regime where a large number of smeared electrons blurs into a continuous charge cloud.\\

\noindent As before, let
\begin{equation*} k: \IR^d \to \IR^d, \, q \mapsto \sigma\, \frac{q}{\lvert q\rvert^d},\end{equation*}
denote the Coulomb kernel. That is, if $\Psi : \IR^d \to \IR$ is a solution of Poisson's equation
\begin{equation*} \Delta \Psi = \mp c \,\rho,\;\; \lim_{\lvert q\rvert\to +\infty} \Psi(q) = 0\end{equation*}
in the sense of 
\begin{equation*}
\Psi(q) = \int \frac{\sigma}{\lvert q-q' \rvert^{d-2}} \; \rho(q') \, \dd^d q', \; d\geq 3,
\end{equation*}
or \begin{equation*} \Psi(q)= -\sigma \int \ln(q-q') \rho(q') \, \dd q',  \, \text{ for } d = 2 \end{equation*} then
\begin{equation*} 
- \nabla \Psi (q) = k*\rho (q) = \sigma \int \frac{q-q'}{\lvert q-q' \rvert^d} \; \rho(q') \, \dd^d q'.
\end{equation*}

\noindent We consider a system of $N$ charges, smeared out by a smooth, non-negative, spherically symmetric form-factor $\chi \in C^\infty_0 (\IR^d)$. For simplicity, we shall assume that $\chi$ satisfies: 

	\begin{enumerate}[i)]
		\item $\mathrm{supp}(\chi) \subseteq \mathrm{B}(1;0) = \lbrace x \in \IR^d : \lVert x \rVert \leq 1 \rbrace$.
		\item $\lVert \chi \rVert_\infty = \sup_{x \in \IR}\;  \chi (x) = 1$. 
		\item $\lVert \chi \rVert_1 = \int \chi(x) \dd x =1$.
	\end{enumerate}

\noindent We call a sequence $(r_N)_{N\in \IN}$ of positive real numbers a \emph{rescaling sequence} if it is monotonously decreasing with $r_1 = 1$ and $\lim\limits_{N \to \infty} r_N = 0$. Given such a rescaling sequence, we define a rescaled form-factor as 
		\begin{equation}\label{rescaled}\chi^{N}(x):= \frac{1}{r_N^{d}}\, \chi\Bigl(\frac{x}{r_N}\Bigr), \; N \in \IN. \end{equation}

\noindent The configuration of the microscopic system is given by $Z(t) = (q_i(t),p_i(t))_{1\leq i \leq N}$, where $q_i(t)$ is the center of mass of particle $i$, and $p_i(t)$ the corresponding momentum at time $t$. The equations of motion in the so called \textit{mean field scaling} read:

\begin{equation} \begin{cases}\label{microscopiceq} \dot{q_i}(t) = p_i(t)\\[1.2ex] \dot{p_i}(t) = K^N(q_i ; q_1,...,q_N)  \end{cases}\end{equation}
\noindent with
\begin{align}\label{RCforce} K^N(q_i ;  q_1,...,q_N) := \frac{1}{N} \sum\limits_{j=1}^N \int \int \chi^N(q_j - y) k(z-y) \chi^N(q_i - z)  \, \dd^d y \, \dd^d z. \end{align}

\noindent The $N$-particle force \eqref{RCforce} can be rewritten in the following way: Given the microscopic density $\mu^N_t = \frac{1}{N} \sum\limits_{i=1}^N \delta_{q_i(t)}\delta_{p_i(t)}$, one checks that
\begin{equation*} K^N(\cdot ; q_1,...,q_N) = \chi^N*k*\chi^N*\rho[\mu^N_t] =: \tilde{k}*\tilde{\rho}[\mu^N_t], \end{equation*}
\noindent where we introduce the notation 
\begin{equation} \tilde{\varphi} := \chi^N*_x\varphi \end{equation}
\noindent for $\varphi$, a measure or measurable function on $\IR^d$ and $*$ denoting the convolution with respect to the space variable. 

Except for the scaling-factor ${N}^{-1}$, these equations describe the regular Coulomb dynamics for smeared charges with form-factor $\chi^N$. They can be understood as the nonrelativistic limit of a Maxwell-Lorentz system of rigid charges (also known as the Abraham model, c.f. \cite[Chs. 2, 13]{Spohn}). The double-convolution results from the fact that the charge enters the interaction-term quadratically; In other words, the charges acting and the charge being acted upon are both smeared out. Note that this system is Hamiltonian for
\begin{equation*} H(q_i,p_i) = \sum\limits_{i=1}^N \frac{1}{2}\, p_i^2 \, + \, \frac{1}{2N} \sum\limits_{i,j} \int \int \chi(y-q_i)\frac{\sigma}{\lvert z-y \rvert^{d-1}} \chi(z-q_j) \dd y\, \dd z, \end{equation*}

\noindent and thus conserves total energy. Note also that this Hamiltonian includes self-interactions.   

\subsection{The regularized Vlasov-Poisson equation} 

For the microscopic model described above, we introduce a corresponding mean field equation: 
\begin{equation}\begin{split}\label{RVP} 
&\partial_t f + p \cdot \nabla_q f + k^N[\rho_t] \cdot \nabla_p f = 0,\\
& k^N[\rho_t](q) := \chi^N*k*\chi^N*\rho_t(q),\\
&\rho_t(q) = \rho[f_t](q) = \int f(t,q,p) \, \dd^d p
\end{split}\end{equation}

\noindent we call this the \textit{regularized Vlasov-Poisson} system with cut-off parameter $r_N$. For $N \to \infty$, the form-factor $\chi^N$ approximates a delta-measure in the sense of distributions and \eqref{RVP} formally reduces to the Vlasov-Poisson system \eqref{VP}.  


\subsection{The method of characteristics}
Let $\nu = (\nu_t)_{t \in [0,T)}$ a continuous family of probability measures on $\IR^d \times \IR^d$ for $T \in \IR^+\cup \lbrace + \infty \rbrace$.  Let $\rho_t[\nu](q)= \int \nu(q,p) \, \dd^d p$ the induced charge distribution on $\IR^d$. We denote by $\varphi^\nu_{t,s} = \bigl(Q^\nu(t,s,q_0,p_0), P^\nu(t,s,q_0,p_0)\bigr) $ the one-particle flow on $\IR^d \times \IR^d$ solving:
\begin{equation}\label{chareq} \begin{cases} 
\frac{\dd}{\dd t} Q = P\\
\frac{\dd}{\dd t} P = \chi^N*k*\chi^N*\rho(Q)\\
Q(s,s,q_0,p_0)= q_0\\
P(s,s,q_0,p_0) = p_0\end{cases}
\end{equation}

\noindent This flow exists and is well-defined since the vector-field is Lipschitz for all $N$. If $f^N(t,q,p)$ is a solution of \eqref{RVP}, it is now straight-forward to check that 
\begin{equation} f^N_t = \varphi^{f^N}_{t,s} \# f^N_s, \;\; \forall 0 \leq s \leq t < T.\end{equation}
\noindent Here,  $\varphi(\cdot) \# f$ denotes the image-measure of $f$ under $\varphi$, defined by  $\varphi \#f(A) = f(\varphi^{-1}(A)) $ for any Borel set $A \subseteq \IR^6$.\\

\noindent Conversely, if $f_t$ is a fixed-point of $(\nu_t) \to \varphi^\nu_t \# f_0$, it is a solution of \eqref{RVP} with initial datum $f_0$. In particular, one observes that $Z(t)=(q_i(t), p_i(t))_{i=1,..,N}$ is a solution of \eqref{microscopiceq} if and only if  $\mu^N[Z(t)]=\frac{1}{N}\sum\limits_{i=1}^N \delta_{q_i(t)}\delta_{p_i(t)}$ solves \eqref{RVP} in the sense of distributions. Basically, our aim is thus to show that this relation carries over to the limit $N \to \infty$.\\

\noindent For the (unregularized) Vlasov-Poisson equation, the corresponding vector-field is not Lipschitz, in general. However, if we assume the existence of a solution $f_t$ with $\rho \in L^\infty([0,T]\times\IR^d)$, the mean field force $k*\rho_t$ does satisfy a Log-Lip bound of the form $\lvert k*\rho_t(x) - k*\rho_t(y)\rvert \leq C \lvert x-y \lvert \lvert \log(\rvert x-y\rvert) \rvert$ (for $\lvert x-y\rvert < \frac{1}{2}$, let's say). This is sufficient to ensure the existence of a characteristic flow $\psi_{t,s}=(Q_{t,s},P_{t,s})$ solving
\begin{equation}\begin{cases}
\frac{\dd}{\dd t} Q_{t,s} = P_{t,s}\\
\frac{\dd}{\dd t} P_{t,s} = k*\rho[f_t](Q_{t,s})\\
Q(s,s,q_0,p_0)= q_0\\
P(s,s,q_0,p_0) = p_0\end{cases}
\end{equation}  
such that $f_t = \psi_{t,s}\# f_s,$ for all $0 \leq s \leq t \leq T$.


\subsection{Existence of Solutions}

\noindent For the regularized mean field equations \eqref{RVP}, all forces are Lipschitz continuous and the solution theory is fairly standard, see e.g. \cite{Dobrushin, BraunHepp}. For the actual Vlasov-Poisson system, the issue is more subtle. Fortunately, in the physically most relevant, 3-dimensional case, we can rely on various results, establishing global existence and uniqueness of (classical) solutions under reasonable conditions on the initial $f_0$ (Pfaffelmoser, 1990 \cite{Pfaffelmoser}, Schaeffer, 1991 \cite{Schaeffer}, Lions and Perthame, 1991 \cite{LionsPerthame}, Horst, 1993 \cite{Horst}). The situation is similar in the 2-dimensional case, treated in Ukai and Okabe, 1978, \cite{Ukai} and Wollmann, 1980 \cite{Wollmann}.\\

\noindent For the rest of the paper, we shall work under the following assumption:

\begin{Assumption}
For $f_0 \in L^1\cap L^\infty(\IR^d \times \IR^d; \IR^+_0)$ there exists a $T^*> 0$ such that the Vlasov-Poisson system (1-3) has a unique solution $f_t$ on $[0, T^*)$ with $f(0, \cdot, \cdot) = f_0$. Moreover, as we consider the sequence of solutions to the regularized equations \eqref{RVP}, the charge density remains bounded uniformly in $N$ and $t$, i.e. $\exists C_0 < + \infty$ such that
\begin{equation}\label{GeneralAssumption}\lVert \rho[f^N_t] \rVert_\infty \leq  C_0, \; \forall t <T^*\; \forall N\in \IN\cup\lbrace +\infty \rbrace, \end{equation}
\noindent where, with a slight abuse of notation, $f^\infty_t :=f_t$. 
\end{Assumption}

\noindent In fact, given a bounded charge density, uniqueness of the solution (in the set of bounded, positive measures) is proven in Loeper, 2006 \cite{Loeper}. Moreover, it is well known that as long as the charge density is bounded, solutions with smooth initial data remain smooth (see e.g. in \cite{Horst3}). 

In dimension $d=3$, the existence result of Lions and Perthame \cite{LionsPerthame} ensures that the above assumption is satisfied for a relatively large class of initial data and $T^*=+\infty$. 

\begin{Theorem}[Lions and Perthame]\label{Thm:LP}\mbox{}\\
	Let $f_0 \geq 0, f_0 \in  L^1(\IR^3\times \IR^3) \cap L^\infty(\IR^3 \times \IR^3)$ satisfy
	
	\begin{equation} \int \lvert p \rvert^m f_0(q, p)  \,\dd q\,  \dd p < + \infty\end{equation}
	for all $m<m_0$ and some $m_0 >3$. 
	\begin{enumerate}[1)]
		\item Then, the Vlasov-Poisson system defined by equations (1--3) has a continuous, bounded solution $f(t, \cdot,\cdot) \in C(\IR^+;L^r(\IR^3\times\IR^3) \cap L^\infty(\IR^+;L^\infty(\IR^3 \times \IR^3))$,  $1 \leq r < \infty$, satisfying
		\begin{equation} \sup\limits_{t \in [0,T]} \int \lvert p \rvert^m  f(t,q, p)  \,\dd p \, \dd p < +\infty, \end{equation}
		for all $T < \infty, m < m_0$.\\
		
		\item If, in fact, $m_0 > 6$ and we assume that $f_0$ satisfies 
		\begin{equation}\label{Assumption} \begin{split}
		\mathrm{supess} \lbrace f_0(q'+p t, p' ) : \lvert q-q'\rvert \leq Rt^2, \lvert p-p' \rvert < Rt \rbrace\\ 
		\in L^\infty\bigl((0,T)\times \IR^d_q; L^1(\IR^3_p)\bigr) \end{split} \end{equation}
		for all $ R >0$ and $T>0$, then there exists $C>0$ such that
		\begin{equation}\label{rhobound}\lVert \rho[f^N_t] \rVert_\infty< C, \; \forall t > 0\; \forall N\in \IN\cup\lbrace \infty \rbrace. \end{equation}
	\end{enumerate}
\end{Theorem}

\noindent Note that Lions and Perthame state \eqref{rhobound} only for $f_t$, though they remark (and it is straightforward to check) that the proof actually yields an upper bound on the charge densities $\rho[f^N_t]$ as one considers a sequence of regularized time-evolutions as, for instance, in \eqref{RVP}.\\

\noindent In higher dimensions, where blow-up might occur, there exists at least some $T^*>0$, depending only on $f_0$, such that \eqref{Assumption} is satisfied, if one assumes that $f_0$ has compact support. This is ensured by the following lemma. 

\begin{Lemma}[Local existence of solutions]
	Let $f_0 \in L^1\cap L^\infty (\IR^{d}\times \IR^3)$ with compact support and $f$ a (local) solution of \eqref{VP} with $f\lvert _{t=0} = f_0$.  Let 
	\begin{align}D(t):= \sup \, \lbrace \lvert q \rvert : \exists  p \in \IR^d: f(t,q,p) \neq 0 \rbrace\\
	R(t):= \sup \, \lbrace \lvert p \rvert : \exists  q \in \IR^d: f(t,q,p) \neq 0 \rbrace
	\end{align}  the diameter of the support in the $q$-, respectively $p$-coordinates. Then there exists a constant $C > 0$ such that
	\begin{align} D(t) &\leq D(0) +  \int\limits_{0}^t R(s) \, \dd s\\
	R(t) &\leq R(0) + C \,  \lVert f_0 \rVert_\infty \lVert f_0 \rVert_1^{1/d} \int\limits_{0}^t R^{d-1}(s) \, \dd s. \end{align}
	These estimates hold independent of $N$ as we consider the sequence $f^N$ of solutions to the regularized equation \eqref{RVP} with $f^N\lvert _{t=0} = f_0$.
\end{Lemma}

\section{Statement of the results}
 Our approximation result for the Vlasov-Poisson dynamics is formulated in terms of (modified) Wasserstein distances. In the context of kinetic equations, the Wasserstein distance was introduced by Dobrushin \cite{Dobrushin}. Here, we shall briefly recall the definition and some basic properties. For further details, we refer the reader to the  book of Villani \cite[Ch. 6]{Villani}.

\begin{Definition}
 \noindent Let $\M(\IR^k)$ the set of probability measures on $\IR^k$ (equipped with its Borel algebra). For given $\mu, \nu \in \M(\IR^k)$ let $\Pi(\mu, \nu)$ be the set of all probability measures $\IR^k \times \IR^k$ with marginal $\mu$ and $\nu$ respectively. 

 \noindent For $p\in [1,\infty)$  we define the \emph{Wasserstein distance} of order $p$ by
 
 \begin{equation}
 W_p(\mu, \nu) := \inf\limits_{\pi \in \Pi(\mu,\nu)} \, \Bigl( \int\limits_{\IR^k\times\IR^k} \lvert x -  y \rvert^p \, \dd \pi(x,y) \, \Bigr)^{1/p}.   
 \end{equation}
 
\noindent Convergence in Wasserstein distance implies, in particular, weak convergence in $\M(\IR^k)$, i.e. 
 	 \begin{equation*} \int \Phi(x)\, \dd \mu_n(x) \to \int \Phi(x)\, \dd \mu(x), \;\;\; n \to \infty, \end{equation*}
 	 for all bonded, continuous functions $\Phi$. Moreover, convergence in $W_p$ implies convergence of the first $p$ moments. $W_p$ satisfies all properties of a metric on $\M(\IR^k)$, except that it may take the value $+\infty$. 
 \end{Definition}	 
 	 
 \noindent An important result is the \emph{Kantorovich-Rubinstein duality}:
 \begin{equation}\begin{split}\label{Kantorovich}W^p_p(\mu, \nu) =  \sup \Bigl\lbrace &\int \Phi_1(x) \, \dd\mu(x) - \int \Phi_2(y) \, \dd\nu(y) : \\ 
 &(\Phi_1, \Phi_2) \in L^1(\mu)\times L^1(\nu), \Phi_1(y) - \Phi_2(x) \leq \lvert x - y \rvert^p \Bigr \rbrace. \end{split} \end{equation}
 
 \noindent A particularly useful case is the first Wasserstein distance, for which the problem reduces further to
  \begin{equation*} W_1(\mu, \nu) = \sup\limits_{\lVert \Phi \rVert_{Lip} \leq 1}\Bigl\lbrace \int \Phi(x) \, \dd \mu(x) -  \int \Phi(x)\, \dd \nu(x) \Bigr\rbrace, \end{equation*}
  where $\lVert \Phi \rVert_{Lip}:= \sup\limits_{x\neq y} \frac{\Phi(x)-\Phi(y)}{\lvert x - y \rvert}$, to be compared with the \emph{bounded Lipschitz distance}
  \begin{equation*} \dbl{\mu, \nu} = \sup \Bigl\lbrace \int \Phi(x) \, \dd \mu(x) -  \int \Phi(x)\, \dd \nu(x)\, ;\; \lVert \Phi \rVert_{Lip}, \lVert \Phi \rVert_\infty \leq 1\Bigr\rbrace.\end{equation*}

\noindent We can now state our precise results in the following theorems. 

	 
\begin{Proposition}[Deterministic Result] \label{Prop:Prop}
	Let $f_0 \in L^1\cap L^\infty(\IR^d\times \IR^d), \, f \geq 0$. Let $(r_N)_{N\in \IN}$ be a rescaling sequence and $f^N_t$ the unique solution of the regularized Vlasov-Poisson equation \eqref{RVP} with $f^N(0, \cdot,\cdot ) = f_0$. Assume that on $[0,T]$ the sequence $(f_N)_N$ satisfies the uniform bound \eqref{GeneralAssumption} on the induced charge densities. Suppose we have a sequence of initial conditions $Z \in \IR^{6N}$ such that 
	\begin{equation}\label{sufficientlyfast} \lim\limits_{N \to \infty} r_N^{-(1+\frac{d}{2}+\epsilon)} W_2(\mu^N_0[Z], f_0) = 0 \end{equation}
	
	\noindent for some $\epsilon >0$. Then it holds that
	\begin{equation}
	\lim\limits_{N \to \infty} r_N^{-(1+\frac{d}{2})} W_2(\mu^N_t[Z], f^N_t) = 0, \; \forall 0 \leq t \leq T.
	\end{equation}
\end{Proposition}

\noindent Since we will also show that $W_2(f^N_t, f_t) = o(r_N^{1-\epsilon})$ (Proposition \ref{Prop:fNtof2}) this establishes a particle approximation of the Vlasov-Poisson equation for initial conditions satisfying \eqref{sufficientlyfast}. 

\begin{Theorem}[Typicality Result]\label{Thm:Thm}
	Let $f_0 \in L^\infty(\IR^d\times \IR^d)$ a probability measure such that the Vlasov-Poisson equation \eqref{VP} has a unique solution on $[0,T^*), T^* \in \IR^+ \cup \lbrace + \infty \rbrace$ with $f(0,\cdot,\cdot)=f_0$. Assume that the sequence $(f_N)_N$ of solutions to the regularized Vlasov-Poisson equation \eqref{RVP} with the same initial data satisfies the uniform bound \eqref{GeneralAssumption} on the induced charge densities. Assume, in addition, that there exists $k > 4$ such that 
	\begin{equation}\label{finitemoment}
	M_k(f_0):= \int (\lvert q \rvert + \lvert p \rvert)^k\, f_0(q,p) \, \dd q\, \dd p < +\infty.
	\end{equation}
	\noindent Suppose that $r_N \geq N^{-\delta}$ with
	\begin{equation*} 
	\delta = \frac{1-\epsilon}{d(2+d +2\epsilon)}, \; \epsilon >0. 
	\end{equation*}
	Then there exist constants $C_1, C_2, C_3$ such that for all $T<T^*$ and $N$ large enough that $r_N \leq \exp[-(\frac{2C_1T+1}{\epsilon})^2]$ it holds that
	\begin{equation}
	\IP_0\Bigl[ \sup\limits_{t\in [0,T]} W_2(\mu^N_t[Z], f_t) > r_N^{1-\epsilon} \Bigr] \leq C_2\bigl(e^{-C_3 N^{\epsilon}} +  N^{1-\frac{k}{2}+\frac{k}{2d}} ),
	\end{equation}
	where the probability $\IP_0$ is defined in terms of the product measure $\otimes^N f_0$ on $(\IR^{d}\times \IR^{d})^N$.
	The constant $C_1$ depends on $d, \chi$ and $C_0$ as in \eqref{GeneralAssumption}, while $C_2, C_3$ depend on $d$, $k$ and $M_k(f_0)$. 
\end{Theorem}

\begin{Remark}\mbox{}
	\begin{enumerate} 
		\item In dimension 3, the necessary cut-off is of order $N^{-\delta}$ with $\delta <\frac{1}{15}$. 
		\item If the finite moment condition \eqref{finitemoment} is replaced by the assumption of a finite exponential moment $\int e^{\gamma \lvert x \rvert^\kappa} \dd f_0(x)$, the rate of convergence becomes exponential, as well. This holds, in particular, for compactly supported $f_0$. 
	\end{enumerate}
\end{Remark}

\subsection{Sketch of the Proof}
We give here a brief sketch of our derivation and the central concepts and ideas on which it is based. 
\begin{enumerate}
	\item  To control the distance between microscopic density and mean field density, we introduce a variant $W^N_2$ of the second Wasserstein distance defined with respect to the $N$-dependent metric: 
	\begin{equation*} d^N\bigl( (q_1,p_1), (q_2, p_2)\bigr) := (1 \vee \sqrt{\lvert\log(r_N)\rvert})\, \lvert q_1 - q_2 \rvert + \lvert p_1-p_2 \rvert.\end{equation*}
	where $ a \vee b := \max \lbrace a, b \rbrace$. 
	\item  We use an estimate from Loeper's proof of uniqueness of weak solutions with bounded density \cite{Loeper} to control the $L^2$-norm of the difference between mean field force and microsocpic force in terms of the quadratic Wasserstein distance. 
	
	\item  The regularization yields a Lipschitz bound on the microscopic force that diverges logarithmically with $N$. In terms of the modified Wasserstein distance, this leads to a Gronwall estimate of the form \begin{equation*}\label{Gronwallestimate} \frac{\dd}{\dd t} W^N_2(\mu^N_t, f^N_t) \leq C \sqrt{\lvert\log(r_N)\rvert}  W^N_2(\mu^N_t, f^N_t). \end{equation*} 
	
	\item The previous bounds can be applied if the (smeared) microscopic charge density $\tilde \rho^\mu = \chi^N*\rho[\mu_t]$ remains bounded uniformly in $N$. We show that this can be assured as long as $W_2(\mu^N_t[Z],f^N_t) = O(r_N^{-(1+d/2)})$. Given a sufficiently fast rate of convergence at $t=0$, i.e. assumption  \eqref{sufficientlyfast}, we conclude with 3. that this bound propagates. 

	\item It remains to check that the constraints so imposed on the initial data are satisfied for typical $Z$, if the initial configuration is chosen randomly according to the product law $\otimes^N f_0$. This is achieved with a recent large deviation estimate found by Fournier and Guilin \cite{Fournier}. This estimate also sets the upper bound on the rate at which $r_N$ can go to zero in the limit $N \to \infty$. 
	
\end{enumerate}

\section{A Gronwall-type Argument}
\noindent Our mean field limit is based on the following stability result by Loeper \cite[Thm. 2.9]{Loeper}, which is proven by methods from the theory of optimal transportation.

\begin{Proposition}[Loeper]\label{Prop:Loeper}
\noindent Let $k$ the Coulomb kernel and $\rho_1, \rho_2 \in L^{1}(\IR^d)\cap L^{\infty}(\IR^d)$ two (probability) densities. Then
\begin{equation} \lVert k*\rho_1 - k*\rho_2 \rVert_2 \leq \bigl[\max\lbrace \lVert \rho_1 \rVert_\infty, \rVert \rho_2 \rVert_\infty \rbrace\bigr]^{1/2} \,W_2(\rho_1, \rho_2). \end{equation}

\end{Proposition}

\noindent Moreover, we require the following estimates on  the mean field force:

\begin{Lemma}\label{Lemma:Lip-Log}
	Let $k$ as before and $\rho \in L^1(\IR^d)\cap L^\infty(\IR^d)$. Then it holds that
	\begin{enumerate}[i)]
		\item $\lVert k * \rho \rVert_\infty \leq   \lvert S^{d-1} \rvert\,\lVert \rho \rVert_\infty + \lVert \rho \rVert_1$.
		\item $\lVert \chi^N * k * \rho \rVert_{Lip} \leq C_{L} (1 \vee \lvert \log(r_N)\rvert) \, \bigl(\lVert \rho \rVert_1 + \lVert \rho \rVert_\infty\bigr)$ 
	\end{enumerate}
	\noindent where we use again the notation $a \vee b := \max \lbrace a , b \rbrace$. $ \lvert S^{d-1} \rvert$ denotes the area of the unit sphere and $C_L$ is a constant depending on $\chi$.
	
\end{Lemma}
\begin{proof} 
	
	i) For the first inequality, we compute
	\begin{align*}\notag \lVert k * \rho \rVert_\infty &\leq \Bigl \lVert \int\limits_{\lvert y \rvert <1} k(y) \rho_t(x-y) \, \mathrm{d}^dy \Bigr \rVert_\infty + \Bigl \lVert \int\limits_{\lvert y \rvert >1} k(y) \rho(x-y) \, \mathrm{d}^dy \Bigr \rVert_\infty\\ \label{forcebound2}
	&\leq \lVert \rho \rVert_\infty \,  \int\limits_{\lvert y \rvert <1} \frac{1}{\lvert y \rvert^{d-1}} \mathrm{d}^d y + \lVert \rho \rVert_1 = \lvert S^{d-1} \rvert \lVert \rho \rVert_\infty + \lVert \rho \rVert_1
	\end{align*}
	
	\noindent ii)We split the expression as
	\begin{align*} \bigl \lVert \nabla ( \chi * k * \rho)  \bigr\rVert_\infty &\leq \bigl\lVert \nabla ( \chi * k\vert_{x \geq r_N^{d+1}} * \rho)  \bigr\rVert_\infty + \bigl\lVert \nabla ( \chi * k\vert_{x < r_N^{d+1}} * \rho)  \bigr\rVert_\infty\\
	&\leq \bigl\lVert \chi^N  \bigr\rVert_1 \, \bigl\lVert \, \nabla k \vert_{x \geq r_N^{d+1}} * \rho  \bigr\rVert_\infty + \bigl\lVert \nabla \chi^N  \bigr\rVert_\infty\, \bigl\lVert k \vert_{x < r_N^{d+1}}  \bigr\rVert_1 \,\bigl\lVert \rho \bigr\rVert_\infty
	\end{align*}
	Now, we have: 
	\begin{align*}\notag \biggl \lvert \nabla k \vert_{x \geq r_N^{d+1}} * \rho \, (x)\biggr\rvert &\leq \int\limits_{\lvert y \rvert \geq r_n^{d+1}} \frac{1}{\lvert y \rvert ^d}\, \rho(x-y)\, \dd^d y \\\notag
	&\leq  \int\limits_{ r_N^{d+1} \leq \lvert y \rvert \leq 1} \frac{1}{\lvert y \rvert ^d} \rho(x-y) \dd^d y + \int\limits_{\lvert y \rvert > 1} \frac{1}{\lvert y \rvert ^d} \rho(x-y) \dd^d y \\
	&\leq (d+1) C\, \lVert  \rho \rVert_\infty \log(r_N^{-1}) + \lVert \rho \rVert_1.
	\end{align*} 
	
	\noindent Furthermore: 
	\begin{equation*}\lVert \nabla \chi^N \rVert_\infty = r_N^{-(d+1)} \lVert \nabla \chi \rVert_\infty \end{equation*}
	and 
	
	\begin{equation*}\bigl \lVert k \vert_{x < r_N^{d+1}} \bigr\rVert_1 = \int\limits_{\lvert y \rvert < r_N^{d+1}} \frac{1}{\lvert y \rvert^{d-1}}\; \dd^d y = \lvert S^{d-1} \rvert \, r_N^{d+1}. \end{equation*}
	Putting everything together, the statement follows. 
\end{proof}  

\noindent For the continuous solutions $f^N_t$ to the (regularized) Vlasov-Poisson equation, the corresponding charge-densities $\rho_t = \rho[f^N_t]$ are bounded by assumption. The challenge is to provide a bound on the microscopic charge density that holds uniformly in $N$, i.e. as the electron radius decreases and the forces become more singular. The idea is to show that as long as $\mu_t^N$ and $f^N_t$ are close as probability measures, the $L^\infty$-norm of $\rho[f^N_t]$ provides a bound on the $L^\infty$-norm of $\tilde\rho[\mu^N_t]$. A simple such estimate was obtained in \cite[Prop. 2.1]{BGV} for the first Wasserstein distance. In view of the general Kantorovich-Rubinstein duality, we generalize this result to Wasserstein distances of higher order.

 
 


\begin{Lemma}\label{Lemma:rhobound} 
	Let $\rho_1, \rho_2$ two probability measures on $\IR^d$ and $\rho_2 \in L^\infty(\IR^d)$. Then:
	\begin{equation} 
		\lVert \tilde{\rho}_1  \rVert_\infty \leq \lvert \mathrm{B}^d(2) \rvert\, \lVert \rho_2 \rVert_\infty +  r_N^{-(p+d)} \,  W_p^p(\rho_1, \rho_2),
	\end{equation}
	where $\mathrm{B}^d(2) \subset \IR^d$ is the $d$-dimensional ball with radius 2. 
\end{Lemma}

\begin{proof}
	For any integrable function $\Phi$, we consider the $c$-conjugate
	\begin{equation*} \Phi^c(y) := \sup\limits_{x} \lbrace \Phi(x) - \lvert x-y \rvert^p \rbrace\end{equation*} 
	This is the smallest function satisfying $\Phi^c(y) \geq \Phi(y)$ and $\Phi(x) - \Phi^c(y) \leq \lvert x-y \rvert^p, \, \forall x,y \in \IR^d$.\\ Now, we write
	\begin{equation*} \begin{split}\tilde{\rho}_1 (x) = r_N^{-(d+p)} \Bigl[\int r_N^{d+p}\chi^N(x-y)  \rho_1(y) \dd y - \int (r_N^{d+p}\chi^N(x-\cdot))^c(z) \rho_1(z) \, \dd z\\ + \int (r_N^{d+p}\chi^N(x-\cdot))^c(z) \, \rho_1(z) \dd z \Bigr] \end{split} \end{equation*}
	By the Kantorovich duality theorem \eqref{Kantorovich} we have
	\begin{equation*} \int r_N^{d+p} \chi^N(x-y)\,  \rho_1(y) \dd y \, - \int (r_N^{d+p}\chi^N(x-\cdot))^c(z)\,  \rho_2(z) \dd z \leq W_p^p(\rho_1, \rho_2). \end{equation*}
	It remains to estimate  
	\begin{equation*} \int (r_N^{d+p} \chi^N(x-\cdot))^c(z)\,  \rho_2(z) \, \dd z. \end{equation*}
	
	\noindent Recalling that $\lVert \chi^N \rVert_\infty = r_N^{-d}$, we find 
	\begin{equation*}(r_N^{d+p} \chi^N(x-\cdot))^c(z) = \sup\limits_{y \in \IR^3} \lbrace  r_N^{d+p} \chi^N(x-y) - \lvert y-z \rvert^p \rbrace \leq r_N^{d+p} \lVert \chi^N \rVert_\infty = r_N^{p}.\end{equation*} 
	Moreover, we observe that 
	\begin{equation} \supp (r_N^{d+p} \chi^N(x-\cdot))^c \subseteq \mathrm{B}(2r_N ; x) := \lbrace z \in \IR^3 : \lvert z -x \rvert \leq 2r_N \rbrace, \end{equation}
	since $ \lvert z- x \rvert >2r_N$ implies $\chi^N(x-y) = 0$, unless $\lvert y-z \rvert \geq r_N$. But then: $r_N^{d+p} \chi^N(x-y) - \lvert y-z \rvert^p \leq r_N^{d+p} r_N^{-d} - r_N^p = 0$. 
	Hence,
	\begin{equation*} \begin{split} \int (r_N^{d+p}\chi^N(x-\cdot))^c(z)  \rho_2(z) \dd z
			\leq \lVert  \rho_2 \rVert_\infty \,r_N^p \, \lvert \mathrm{B}(2 r_N; x) \rvert
			\leq 2^d \lvert \mathrm{B}^d(1)\rvert \, \lVert  \rho_2 \rVert_\infty\, r_N^{d+p}. \end{split}\end{equation*}
	
	\noindent In total, we find
	\begin{equation*} \lVert \tilde{\rho}_1 \rVert_\infty \leq r_N^{-(p+d)} \, W_p^p(\rho_1, \rho_2) + \lvert B^d(2)\rvert \lVert  \rho_2 \rVert_\infty\end{equation*}
	as announced.
\end{proof}

\noindent We shall apply the previous Lemma to $\rho_1 := \rho[\mu^N_t(Z)]$ and $\rho_2 := \rho[f^N_t]$ using $\lVert  \rho[f^N_t] \rVert \leq C_\rho$ and $W_2(\rho[\mu^N_t(Z)], \rho[f^N_t]) \leq W_2(\mu^N_t(Z), f^N_t)$ to get a bound on the (smeared) microscopic charge density.\\

\noindent Finally, we need the following inequalities for the smeared densities. 

\begin{Lemma}\label{renormweaknorm}\mbox{}\\
			Let $\chi \in C^\infty_0 (\IR^d)$, $(r_N)_N$ a rescaling sequence and $\chi^N$ the rescaled form-factor as defined in \eqref{rescaled}. Let $\mu, \nu \in \M(\IR^d)$ and $\tilde\nu = \chi^N*_x\nu$ etc. Then we have for  $1\leq p<\infty$:
			\begin{enumerate}[i)]
		\item $W_p(\tilde{\nu}, \nu) \leq r_N$
			\item $W_p(\tilde\mu, \tilde\nu) \leq W_p(\mu, \nu)$.
			\end{enumerate}
	
\end{Lemma}
		
		\begin{proof}
			i) Define $\pi(x,y) := \nu(x) \chi^N(x-y)$ and observe that $\int \dd x \, \pi(x,y) = \tilde\nu(y), \; \int \dd y \, \pi(x,y) = \nu(x)$, hence $\pi \in \Pi(\tilde{\nu},\nu)$. $\pi$ has support in $\lbrace \lvert x -y \rvert < r_N \rbrace$. Thus, we conclude
			\begin{align*}  W_p(\tilde{\nu}, \nu) &= \inf\limits_{\pi' \in \Pi(\nu, \tilde{\nu})} \, \Bigl( \int\limits_{\IR^d\times\IR^d} \lvert x - y \rvert^p \, \dd \pi'(x,y) \, \Bigr)^{1/p}\\
			&\leq \Bigl( \int\limits_{\IR^d\times\IR^d} \lvert x - y \rvert^p \, \dd \pi(x,y) \, \Bigr)^{1/p} \leq r_N.\end{align*}
	
\noindent ii) In view of the Kantorovich duality \eqref{Kantorovich}, we find for $(\Phi_1, \Phi_2) \in L^1(\mu)\times L^1(\nu)$ with $\Phi_1(y) - \Phi_2(x) \leq \lvert x - y \rvert^p$:
				\begin{align*} 
				\int \Phi_1(x)\,  \dd \tilde\mu(x) - \int \Phi_2(y)\, \dd \tilde\nu(y) 
				= \int (\chi *\Phi_1) (x)\, \dd \mu(x)  - \int (\chi*\Phi_2) (y)\, \dd \nu(y)
				\end{align*}
				
				\noindent But $\chi *\Phi_1$ and $\chi*\Phi_2$ also satisfy
				\begin{align*}
				&\bigl\lvert \chi* \Phi_1(x) - \chi *\Phi_2 (y) \bigr\rvert =\Bigl\lvert \int \chi(z) \Phi_1(x-z)\, \dd z  -  \int \chi(z) \Phi_2(y-z)\, \dd z  \Bigr\rvert  \\
				\leq & \int \chi(z)\, \bigl\lvert \Phi_1 (x-z) -  \Phi_2(y-z)\bigr\rvert\, \dd z \leq \int \chi(z)\, \lvert x-y \rvert^p\, \dd z =   \lvert x-y \rvert^p. 
				\end{align*}
				
				\noindent Hence, we have
				\begin{equation*} 
				\int \Phi_1\,  \dd \tilde\mu - \int \Phi_2\, \dd \tilde\nu = \int \tilde \Phi_1\,  \dd \mu - \int \tilde \Phi_2\, \dd \nu \leq W_p(\mu, \nu)
				\end{equation*}
				\noindent and taking the supremum over all $(\Phi_1, \Phi_2)$ yields the desired inequality.
				
			\end{proof}
			
\subsection{ Modified Wasserstein distance}		

\noindent As we want to establish a Gronwall inequality for the distance between empirical density and Vlasov density, we aim for a bound of the form: 
\begin{equation*}
\mathrm{dist}(\mu^N_{t+\Delta t},f^N_{t+\Delta t}) - \mathrm{dist}(\mu^N_{t},f^N_{t}) \propto  \mathrm{dist}(\mu^N_{t},f^N_{t}) \, \Delta t + o(\Delta t).
\end{equation*}
The choice of a metric, giving precise meaning to $\mathrm{dist}(\mu^N_{t},f^N_{t})$, is thus a balancing act. While a stronger metric is, in general, more difficult to control, it also yields stronger bounds as it appears on the right hand side of the Gronwall estimate. 

If we compare the characteristic flow of the mean field dynamics with the flow corresponding to the ``true'', i.e. microscopic, dynamics, the growth in the \textit{spatial} distance is trivially bounded by the distance of the respective momenta. The only problem lies in controlling fluctuations in the force, i.e. the growth of the distance in \textit{momentum} space. The idea, first employed in \cite{PeterDustin}, is thus to be more rigid on deviations in the $q$-coordinates, weighing them with an appropriate $N$-dependent factor, and use this to obtain better control on the forces. 

\begin{Definition}
	Let $(r_N)_{N\in\IN}$ be a rescaling sequence. On $\IR^d\times\IR^d$ we introduce the ($N$-dependent) metric: 
	\begin{equation} d^N\bigl( (q_1,p_1), (q_2, p_2)\bigr) := (1 \vee \sqrt{\lvert\log(r_N)\rvert})\, \lvert q_1 - q_2 \rvert + \lvert p_1-p_2 \rvert.\end{equation}
   Now let $W^N_p(\cdot, \cdot)$ be the p'th Wasserstein metric with respect to $d^N$, i.e.: 
	\begin{equation}\label{Def:WN}
		W^N_p(\mu, \nu) := \inf\limits_{\pi \in \Pi(\mu,\nu)} \, \Bigl( \int\limits_{\IR^d \times \IR^d} d^N(x, y)^p \, \dd \pi(x,y) \, \Bigr)^{1/p}.   
	\end{equation}
	
	\noindent Note that $W_p(\mu, \nu) \leq W^N_p(\mu, \nu) \leq (1 \vee \sqrt{\lvert\log(r_N)\rvert})\, W_p(\mu, \nu)$, $\forall \mu, \nu \in \M(\IR^d\times\IR^d)$. Finally, we define
	\begin{equation} W^*(\mu, \nu) := \min \Bigl\lbrace 1, \, r_N^{ -(1+\frac{d}{2})} \, W^N_2(\mu, \nu) \Bigr\rbrace. \end{equation}
	
\end{Definition} 

\noindent Obviously, convergence with respect to $W^*$ is much stronger than convergence with respect to $W_2$. Concretely, we have for any sequence $(\nu_N)_{N \in \IN}$ and $\nu \in \M(\IR^d \times \IR^d)$: 
\begin{equation*} W^*(\nu_N, \nu) \to 0 \Rightarrow W_2(\nu_N, \nu) = {o}\bigl( r_N^{1+\frac{d}{2}} \bigr). \end{equation*}

\subsection{Deterministic result}

 \noindent We now come to the central part of our argument:


\begin{proof}[\emph{\textbf{Proof of Proposition \ref{Prop:Prop}}}]
	Let $N \in \IN$ and $\pi_0 \in \Pi(\mu^N_0, f_0)$. Let $\varphi^\mu_t = (Q^\mu_t, P^\mu_t)$ and $\varphi^f_t=(Q^f_t, P^f_t)$  the flow induced by the characteristic equation \eqref{chareq} for $\mu_t^N$ and $f^N_t$, respectively. For any $t \in [0,T], \, T <T^*$, define the  ($N$-dependent) measure $\pi_t$ on $\IR^{6N} \times \IR^{6N}$ by $\pi_t = (\varphi^\mu_t, \varphi^f_t)\# \pi_0$. Then $\pi_t \in \Pi(\mu^N_t, f_t), \, \forall t \in [0,T]$. We set
	\begin{equation*}\begin{split}
			D(t):= & \Bigl[ \int\limits_{\IR^6\times\IR^6} d^N(x,y)^2 \; \dd \pi_t(x,y)\Bigr]^{1/2}\\
			= &\Bigl[ \int\limits_{\IR^{6} \times \IR^{6}} \Bigl((1 \vee \sqrt{\lvert\log(r_N)\rvert})\, \lvert x^1-y^1 \rvert + \, \lvert x^2-y^2 \rvert \Bigr)^2 \; \dd \pi_t(x,y)\Bigr]^{1/2}\\ 
			=&\Bigl[ \int\limits_{\IR^{6} \times \IR^{6}} \Bigl((1 \vee \sqrt{\lvert\log(r_N)\rvert})\, \lvert Q^\mu_t(x)-Q^f_t(y) \rvert + \, \lvert P^\mu_t(x)-P^f_t(y) \rvert \Bigr)^2 \;\dd \pi_0(x,y)\Bigr]^{1/2}.
		\end{split}\end{equation*}
		\noindent Note that $W^N_2(\mu^N_t, f^N_t) < D(t)$ for any $\pi_0 \in \Pi(f_0, f_0)$. Now we consider: 
		\begin{equation} D^*(t) := \min\bigl\lbrace 1 , \, r_N^{ -(1+\frac{d}{2})}\,D (t) \bigr\rbrace. \end{equation}
		
		\noindent Obviously, $\frac{\dd}{\dd t}D^*(t) \leq 0$ whenever $D(t) \geq r_N^{1+\frac{d}{2}}$ since $D^*(t)$ is already maximal. For $D(t) < r_N^{1+\frac{d}{2}}$, we compute: 
		\begin{align*}\notag
			\frac{\dd}{\dd t} D^2(t) &=\\ 
			2\, \int &\Bigl((1 \vee \sqrt{\lvert\log(r_N)\rvert})\, \lvert Q^\mu_t(x)-Q^f_t(y) \rvert +  \lvert P^\mu_t(x)-P^f_t(y) \rvert \Bigr)\cdot\\
			&\Bigl((1 \vee \sqrt{\lvert\log(r_N)\rvert})\, \lvert P^\mu_t(x)-P^f_t(y) \rvert +  \, \bigl\lvert \tilde k*\tilde\rho^\mu_t (Q^\mu_t(x)) - \tilde k*\tilde \rho^f_t(Q^f_t(y)) \bigr\rvert\Bigr)\,  \dd \pi_0(x,y).
		\end{align*}
		
		\noindent The interesting term to control is the interaction term
		\begin{align}\notag & \bigl\lvert \tilde k*\tilde\rho^\mu_t \, (Q^\mu_t(x)) - \tilde k*\tilde \rho^f_t \,(Q^f_t(y)) \bigr\rvert \\\label{term1}
			&\leq \bigl\lvert \tilde k*\tilde\rho^\mu_t \, (Q^\mu_t(x)) - \tilde k*\tilde \rho^\mu_t \,(Q^f_t(y)) \bigr\rvert \\
			&+ \bigl\lvert \tilde k*\tilde\rho^\mu_t \, (Q^f_t(y)) - \tilde k*\tilde \rho^f_t \,(Q^f_t(y)) \bigr\rvert  
		\end{align}
		
		\noindent We begin with \eqref{term1} and find with Lemma \ref{Lemma:Lip-Log}:
		\begin{equation}\begin{split} 
				& \bigl\lvert \tilde k*\tilde\rho^\mu_t \, (Q^\mu_t(x)) - \tilde k*\tilde \rho^\mu_t \,(Q^f_t(y)) \bigr\rvert \\
				& \leq  C_L  (1 \vee \lvert \log(r_N)\rvert) (1+ \lVert \rho^\mu_t \rVert_\infty)\, \bigl\lvert Q^\mu_t(x) - Q^f_t(y) \bigr\rvert 
			\end{split}
		\end{equation}

		\noindent Hence, we have
		\begin{equation}
			\frac{\dd}{\dd t} D^2(t) \leq J_1(t) + J_2(t)
		\end{equation}
		with
		\begin{multline}\label{DJ1}
				J_1(t) :=  2\, \int  \dd \pi_0(x,y) \Bigl((1 \vee \sqrt{\lvert\log(r_N)\rvert})\, \lvert Q^\mu_t(x)-Q^f_t(y) \rvert +  \lvert P^\mu_t(x)-P^f_t(y) \rvert \Bigr) \cdot \\
				\Bigl((1 \vee \sqrt{\lvert\log(r_N)\rvert}) \lvert P^\mu_t(x)-P^f_t(y) \rvert + C_L  (1 \vee \lvert \log(r_N)\rvert) (1+ \lVert \rho^\mu_t \rVert_\infty) \bigl\lvert Q^\mu_t(x) - Q^f_t(y) \bigr\rvert \Bigr) 
		\end{multline}
		\begin{equation}\begin{split}\label{DJ2}
				J_2(t) := 2\, \int \Bigl((1 \vee \sqrt{\lvert\log(r_N)\rvert})\, \lvert Q^\mu_t(x)-Q^f_t(y) \rvert +  \lvert P^\mu_t(x)-P^f_t(y) \rvert \Bigr) \cdot\\
				\bigl\lvert \tilde k*\tilde\rho^\mu_t \, (Q^f_t(y)) - \tilde k*\tilde \rho^f_t \,(Q^f_t(y)) \bigr\rvert   \,  \dd \pi_0(x,y)
			\end{split}\end{equation}
			
			\noindent Now we observe that
			\begin{equation} J_1(t) \leq  C_L  (1 \vee \lvert \log(r_N)\rvert) (1+ \lVert \rho^\mu_t \rVert_\infty) D^2(t),\end{equation}
			while for the second term, we find with H{\"o}lders inequality
			\begin{align}\notag
				J_2&(t) \leq \\\label{J21}
				2&\Bigl[\int \Bigl((1 \vee \sqrt{\lvert\log(r_N)\rvert})\, \lvert Q^\mu_t(x)-Q^f_t(y) \rvert +  \lvert P^\mu_t(x)-P^f_t(y) \rvert \Bigr)^2  \dd \pi_0(x,y)\Bigr]^{1/2}\\\label{J22}
				&\Bigl[\int\bigl\lvert \tilde k*\tilde\rho^\mu_t \, (Q^f_t(y)) - \tilde k*\tilde \rho^f_t \,(Q^f_t(y)) \bigr\rvert^2   \,  \dd \pi_0(x,y)\Bigr]^{1/2}.
			\end{align}
			We identify \eqref{J21} as $D(t)$, while for \eqref{J22} we get
			\begin{align}\notag 
				& \Bigl[\int \bigl\lvert \tilde k*(\tilde\rho^\mu_t - \tilde\rho^f_t \bigr) \bigl(Q^f_t(y) \bigr) \bigr\rvert^2\;\dd \pi_0(x,y)\Bigr]^{1/2}\\\notag
				 & =  \Bigl[\int \bigl\lvert \tilde k*(\tilde\rho^\mu_t - \tilde\rho^f_t \bigr) \bigl(Q_0(y)\bigr) \bigr\rvert^2\;\dd \pi_t(x,y)\Bigr]^{1/2}\\\notag
				& \leq \Bigl[\int  \bigl(\tilde k*\tilde \rho^\mu_t - \tilde k*\tilde \rho^f_t \bigr)^2 \,f(t,y)\,\dd^{2d} y )\Bigr]^{1/2}\\\notag
				& =  \Bigl[\int  \bigl(\tilde k*\tilde\rho^\mu_t - \tilde k*\tilde\rho^f_t \bigr)^2(q) \, \rho^f_t(q)\,\dd^d q )\Bigr]^{1/2}\\[1.2ex]\label{term2}
			& \leq \lVert \rho^f_t \rVert^{1/2}_\infty \lVert \tilde k*(\tilde \rho^\mu_t - \tilde \rho^f_t) \rVert_2 \leq C_0^{1/2}\lVert k*(\tilde \rho^\mu_t - \tilde \rho^f_t) \rVert_2
			\end{align}
			
			\noindent From Lemma \ref{Lemma:rhobound}, we know that as long as $D(t) \leq r_N^{1+\frac{d}{2}}$, i.e. $D^*(t) \leq 1$, the microscopic charge density is bounded as
			\begin{equation}\begin{split}   \lVert \rho^\mu_t \rVert_\infty \leq &\lvert B^d(2) \rvert \lVert\rho[{f}^N_t] \rVert_\infty +  r_N^{-(d+2)} \,D^2(t)\\[1.2ex]
					\leq & \lvert B^d(2) \rvert \sup_{N \in \IN} \lVert\rho[{f}^N_t] \rVert_\infty +  1\\
					\leq & \lvert B^d(2) \rvert C_0 + 1 =: C_\rho,
				\end{split}
			\end{equation} 
			Note that this bound holds independent of $N$. Hence, we can use Loeper's stability result, Proposition \ref{Prop:Loeper}, in \eqref{term2} and get:
			\begin{align}\label{L2forceestimate} \lVert k*(\tilde\rho^\mu_t - \tilde \rho^f_t )\rVert_2 \leq \bigl[\max\lbrace \lVert \tilde\rho_t^\mu \rVert_\infty, \rVert \tilde\rho_t^f \rVert_\infty \rbrace\bigr]^{\frac{1}{2}} \, W_2(\tilde \rho^\mu_t, \tilde\rho_t^f) \leq C_\rho^{\frac{1}{2}} D(t).\end{align}
		\noindent Putting everything together and setting $C_1:=2C_\rho C_L$, we have
			
			\begin{equation*} \frac{\dd}{\dd t} D^2(t) \leq 2\, C_1 (1 \vee \sqrt{\lvert\log(r_N)\rvert})\, D^2(t) \end{equation*}
			or, after dividing by $2D(t)$ and multiplying both sides by $r_N^{-(1+\frac{d}{2})}$,
			\begin{equation*}
				\frac{\dd}{\dd t} D^*(t) \leq C_1( 1 \vee \sqrt{\lvert\log(r_N)\rvert}) D^*(t).\end{equation*}
			
			\noindent By an application of Gronwall's Lemma, we conclude that:
			
			\begin{equation*} D^*(t) \leq D^*(0) \, e^{t\, C_1( \sqrt{\lvert\log(r_N)\rvert}+1)}.\end{equation*}
			
			\noindent Finally, taking on the right hand side the infimum over all $\pi_0 \in \Pi(\mu^N_0, f_0)$, $D^*(0)$ becomes $W^*(\mu^N_0[Z], f_0)$ and we get for all $t \in T$:
			
			\begin{equation} W^*(\mu_t^N, f^N_t) \leq  W^*(\mu^N_0,f_0)\, e^{t\, C_1 (\sqrt{\lvert\log(r_N)\rvert}+1)}.\end{equation}
			
			\noindent If there exists an $\epsilon > 0$ such that $\lim\limits_{N\to \infty} \frac{W_2(\mu^N_0,f_0)}{r_N^{1+d/2+\epsilon}} = 0$, the right hand side converges to $0$, so that, in particular, $\lim \limits_{N \to \infty} r_N^{1+ \frac{d}{2}} W_2(\mu^N_t,f^N_t) = 0.$ 
			
		\end{proof}

\noindent To show convergence to solutions of the (unregularized) Vlasov-Poisson equation, we also require the following:

\begin{Proposition}\label{Prop:fNtof2}
	Let $f_0$ satisfy the assumptions of Proposition \ref{Prop:Prop}. Let $f^N_t$ and $f_t$ be the solution of the regularized, respectively the proper Vlasov-Poisson equation with initial data $f_0$. Then:
	\begin{equation} W_2(f^N_t, f_t) \leq  r_N\, e^{tC_1(\sqrt{\lvert \log (r_N)\rvert}+1)}.\end{equation}
\end{Proposition}
\begin{proof} 
	Let $\rho^N_t:=\rho[f^N_t]$ and $\rho^\infty_t:=\rho[f_t]$ be the charge density induced by $f^N_t$ and $f_t$, respectively. Let $\varphi_t^N = (Q_t^N, P_t^N)$ the characteristic flow of $f^N_t$ and $\psi_t=(Q_t, P_t)$ the characteristic flow of $f_t$. We consider $\pi_0(x,y) := f_0(x)\delta(x-y) \in \Pi(f_0,f_0)$, which is already the optimal coupling yielding $W^N_2(f_t^N, f_t)\lvert_{t=0}=W^N_2(f_0,f_0) = 0$ and set $\pi_t = (\varphi^N_t, \psi_t)\# \pi_0 \in \Pi(f^N_t,f_t)$. As above, we define
	\begin{equation}
	D(t):=\Bigl[ \int\limits_{\IR^{6} \times \IR^{6}} \Bigl((1 \vee \sqrt{\lvert\log(r_N)\rvert})\, \lvert x^1-y^1 \rvert + \, \lvert x^2-y^2 \rvert \Bigr)^2 \; \dd \pi_t(x,y)\Bigr]^{1/2}
	\end{equation}
	
	\noindent and compute 
	\begin{align*}\notag
	&\frac{\dd}{\dd t} D^2(t) 
	\leq  2\, \int \Bigl((1 \vee \sqrt{\lvert\log(r_N)\rvert})\, \lvert Q^N(t,x)-Q(t,y) \rvert +  \lvert P^N(t,x)-P(t,y) \rvert \Bigr)\\
	&\Bigl((1 \vee \sqrt{\lvert\log(r_N)\rvert}) \, \lvert P^N(t,x)-P(t,y) \rvert +  \, \bigl\lvert \tilde k*\tilde\rho^N_t (Q^N(x)) -  k* \rho^f_t(Q_t(y)) \bigr\rvert\Bigr)\,  \dd \pi_0(x,y)
	\end{align*}
	
	\noindent The proof proceeds analogous to Proposition \ref{Prop:Prop}, simplified by the fact that the charge densities remain bounded by assumption. The only noteworthy difference is in eq. \eqref{L2forceestimate}. Observing that $\tilde k * \tilde \rho = k * \tilde{\tilde \rho} = k * (\chi^N*\chi^N*\rho)$, we use Lemma \ref{renormweaknorm} to conclude:
	\begin{equation} W_2(\tilde{\tilde{\rho}}_t^N, \rho_t) \leq W_2(\rho^N_t,\rho_t) + 2 r_N \leq W_2(f^N_t, f_t) + 2 r_N \leq D(t) + 2 r_N \end{equation}
	
	\noindent In total, we find:
	\begin{align*}\frac{\dd}{\dd t} D^2(t) \leq 2 C_0 C_L\,( 1 \vee \sqrt{\lvert \log{r_N}\rvert}) \, D^2(t) + 2C_0 D(t) ( D(t) +2 r_N) \end{align*}
	\noindent or
	\begin{equation*} 
	\frac{\dd}{\dd t} D(t) \leq  C_1( \sqrt{\lvert \log{r_N}\rvert} +1)\, D(t) + 2 C_0 r_N 
	\end{equation*}
	with $C_1 > 2C_0 C_L$ as defined in the previous proof. Using Gronwall's inequality and the fact that $D(0)=0$, we have
	\begin{equation*}W_2(f^N_t, f_t) \leq  D(t) \leq  r_N \,e^{tC_1( \sqrt{\lvert \log r_N \rvert} +1)}, \end{equation*}
	from which the desired statement follows. 
	
\end{proof}

\subsection{Typicality}

To complete the proof of Theorem \ref{Thm:Thm}, it remains to show that the assumptions of Proposition \ref{Prop:Prop} are satisfied for \textit{typical} initial conditions, i.e. with probability approaching one as $N$ tends to infinity. It is a classical result that if $Z_1, ..., Z_N$ are i.i.d. with law $f$, their empirical density $\mu^N[Z]=\frac{1}{N} \sum\limits_{i=1}^N \delta_{Z_i}$ goes to $f$ in probability. Establishing quantitative bounds on large deviations (concentration estimates) is, however, a longstanding problem in probability theory with a vast amount of literature. To our knowledge, one of the first paper to address this question in the context of Wasserstein metrics was Bolley, Guillin, Villani, 2007 \cite{BGV}. Subsequently, other authors have derived stronger concentration estimates, see, in particular, \cite{Boissard} and \cite{Dereich}. Very recently, great progress has been made in the paper of Fournier and Guillin, 2014, which considerably improves upon previous results, both in strength and generality \cite{Fournier}. We cite now their concentration estimates and apply them to conclude the proof of our main theorem. 

\begin{Theorem}[Fournier and Guillin]\label{Fournier}\mbox{}\\
	\noindent Let $f$ be a probability measure on $\IR^n$ such that $\exists k>2p$:
	\begin{equation*} 
		M_k(f):=	\int\limits_{\IR^k} \lvert x \rvert^k   \dd f(x) < + \infty.
	\end{equation*}
	Let $(Z_i)_{i=1,...,N}$ be a sample of independent variables, distributed according to the law $f$ and consider $\mu^N[Z]:= \sum\limits_{i=1}^N \delta_{Z_i}$. Then, for any $\epsilon > 0$ there exist constants $c,C$ depending only on $k, M_k(f)$ and $\epsilon$ such that for all $N \geq 1$ and $\xi >0$:
	\begin{equation*} 
		\IP_0\Bigl(W^p_p(\mu^N, f) > \xi \Bigr) \leq C N (N\xi)^{-\frac{k-\epsilon}{p}} +  C\mathds{1}_{\xi \leq 1}\, a(N,\xi)\end{equation*}
		with \begin{equation}	\label{aNxi}
	a(N,\xi):=	
		\begin{cases}\exp(-cN\xi^2) & \text{if } p > n/2\\
		 \exp(-cN(\frac{\xi}{\ln(2+1/\xi)})^2) &  \text{if } p = n/2\\
		\exp(-cN\xi^{k/p}) &\text{if } p \in [1,n/2). \end{cases} 
	\end{equation}
\end{Theorem}

\noindent We now apply this result to conclude the proof of our main theorem and establish an upper bound on $r_N$.
\begin{proof}[\textbf{\emph{Proof of Theorem} \ref{Thm:Thm}}]
Let $r_N \geq N^{-\delta}$ and $\epsilon > 0$. Let $\mathcal{A} \subseteq \IR^{2d}$ be the ($N$-dependent)  set defined by 
\begin{equation}Z \in \mathcal{A} \iff W_2(\mu^N_0[Z], f_0) > r_N^{1+\frac{d}{2} + \epsilon}.\end{equation}

\noindent We apply the previous in $n=2d$ dimensions with $\xi = N^{-\delta(2 + d + 2\epsilon)} \leq  r_N^{2(1+\frac{d}{2} + \epsilon)}$ and the finite moment assumption \eqref{finitemoment}. We find:
\begin{equation*} 
	\IP_0(\mathcal{A}) \leq C \Bigl(\exp(-cNN^{-\delta(2 + d + 2\epsilon)d})  +N^{1 -\frac{k - \epsilon}{2} (1 - \delta(2 + d + 2\epsilon))}\Bigr).\end{equation*}
Where the probability is defined with respect to $\otimes^N f_0$. Choosing 
\begin{equation} \delta = \frac{1 - \epsilon}{(2 + d + 2\epsilon)d} \end{equation} we have
\begin{equation*}
	\IP_0(\mathcal{A}) \leq C \bigl(\exp(-cN^\epsilon)  +N^{1 -\frac{k}{2}+\frac{k}{2d}}\bigr) \to 0 , \; N \to \infty.
\end{equation*}
\noindent For the typical initial conditions $Z \in \mathcal{A}^c$, it holds according to Proposition \ref{Prop:Prop} that for all $ t\leq T$,
\begin{align}\notag W^*(\mu^N_t, f^N_t) &\leq  W^*(\mu^N_0,f_0)\, e^{t\, C_1 (\sqrt{\lvert\log(r_N)\rvert}+1)}\\\notag
	&\leq (1 \vee \sqrt{\lvert \log(r_N)\rvert})\, r_N^{-(1+\frac{d}{2})}W_2(\mu^N_0,f_0)  \, e^{t\, C_1 (\sqrt{\lvert \log(r_N)\rvert}+1)} \\\label{besmallerone}
	& \leq  (1 \vee \sqrt{\lvert \log(r_N) \rvert})\,r_N^{\epsilon} \, e^{T\, C_1 (\sqrt{\lvert\log(r_N)\rvert}+1)}.\end{align}
Observing that $e^{\sqrt{\lvert \log r_N \rvert}} = \bigl(e^{- \log r_N})^{\frac{1}{ \sqrt{\lvert \log r_N \rvert }}} = (r_N)^{\frac{-1}{\sqrt{\lvert \log r_N \rvert}}}$, there exists $N_0 \in \IN$ such that  $\eqref{besmallerone} < 1$ for all $N \geq N_0$. More precisely, it suffices to choose $N_0$ large enough that $r_{N_0} < e^{-(\frac{2C_1T +1}{\epsilon})^2}$. Then we find:  
\begin{equation}W^*(\mu^N_t, f^N_t) < 1 \Rightarrow W_2(\mu^N_t, f^N_t) < r_N^{1+\frac{d}{2}} W^*(\mu^N_t, f^N_t) <r_N^{1+\frac{d}{2}}. \end{equation}

\noindent Now we recall from Proposition \ref{Prop:fNtof2}:
\begin{equation*} W_2(f^N_t, f_t) \leq r_N\, e^{tC_1(\sqrt{\lvert \log (r_N)\rvert}+1)}\end{equation*}
which is smaller than $\frac{1}{2} r_N^{1-\epsilon}$ for $N \geq N_0$. We conclude the proof by noting that
\begin{equation*}W_2(\mu^N_t[Z], f_t) \leq W_2(\mu^N_t[Z], f^N_t)  +  W_2( f^N_t, f_t) \leq   r_N^{1+\frac{d}{2}} + \frac{1}{2} r_N^{1-\epsilon} \leq  r_N^{1-\epsilon}, \end{equation*} 
for all $Z \in \mathcal{A}^c, N \geq N_0$ and $t \in [0,T]$. 

\end{proof}

\newpage
\bibliography{VPlit}
\bibliographystyle{plain}

\end{document}